\documentclass[12pt]{article}
\usepackage{amsmath}
\usepackage{amsfonts}
\usepackage{amssymb}
\usepackage{amsthm}
\usepackage{graphicx}
\usepackage{fullpage}
\usepackage{setspace}
\usepackage{verbatim}
\usepackage{natbib}
\usepackage{appendix}
\usepackage{rotating}
\usepackage{multirow}

\newtheorem{theorem}{Theorem}[section]

\newtheorem{lemma}{Lemma}[section]
\newtheorem{corollary}{Corollary}[section]

\theoremstyle{definition}

\DeclareMathOperator{\FDR}{FDR}
\DeclareMathOperator{\FDP}{FDP}

\title{False Discovery Rate Adjustments for Average Significance Level
  Controlling Tests}
\author{Timothy B. Armstrong\thanks{email: timothy.armstrong@usc.edu.  Support
    from National Science Foundation Grant SES-2049765 is gratefully
    acknowledged.}  \\
    University of Southern California}
\date{\today}

\begin{document}

\maketitle

\begin{abstract}
  Multiple testing adjustments, such as the \citet{benjamini_controlling_1995}
  step-up procedure for controlling the false discovery rate (FDR), are
  typically applied to families of tests that control significance level in the
  classical sense: for each individual test, the probability of false rejection
  is no greater than the nominal level.
  In this paper, we consider tests that satisfy only a weaker notion of
  significance level control, in which the probability of false rejection need
  only be controlled on average over the hypotheses.
  We find that the \citet{benjamini_controlling_1995} step-up procedure
  still controls FDR in the asymptotic regime with many weakly dependent
  $p$-values and an increasing number of rejections, and that certain
  adjustments for dependent
  $p$-values such as the \citet{benjamini_control_2001} procedure continue to
  yield FDR control in finite samples.
  Our results open the door to FDR controlling procedures in
  nonparametric and high dimensional settings where weakening the notion of
  inference
  may allow for power improvements.
\end{abstract}

\section{Introduction}

Consider testing $m$ hypotheses $H_1,\ldots,H_m$.  Let
$\mathcal{H}_0\subseteq\{1,\ldots,m\}$ denote the set of true null hypotheses.
Given $p$-values $p_1,\ldots,p_m$ for each of the hypotheses,
we wish to form a multiple testing procedure which decides on a subset of
hypotheses to reject.
A common starting point for multiple testing procedures proposed in the
literature is to assume that the $p$-values are formed from tests that control
significance level in the classical sense, which implies
\begin{align}
  \label{eq:level_alpha}
  \text{for all }t\in [0,1]\text{ and }i\in\mathcal{H}_0,\, P(p_i\le t) \le t.
\end{align}
One then adjusts the critical value so that some notion of multiple testing
error, such as the false discovery rate (FDR), is controlled (see formal
definitions below).

In this paper, we explore the possibility of forming FDR controlling multiple
testing procedures from tests that satisfy a weaker
\emph{average significance level} control criterion:
\begin{align}
  \label{eq:average_level_alpha}
  \text{for all }t\in [0,1], \,\frac{1}{m}\sum_{i\in \mathcal{H}_0} P(p_i\le t) \le t.
\end{align}
Such tests can be formed from confidence intervals (CIs) that weaken the
classical definition of a CI by requiring coverage only on average over the
reported CIs.  Letting $CI_1(t),\ldots,CI_m(t)$ be CIs for parameters $\theta_1,\ldots,\theta_m$ with nominal $100\cdot
(1-t)\%$ coverage, the \emph{average coverage} criterion requires
\begin{align}
  \label{eq:average_coverage}
  \frac{1}{m}\sum_{i=1}^m P(\theta_i\notin CI_i(t))\le t.
\end{align}
Given null hypotheses $H_i: \theta_i=\theta_{0,i}$,
$p$-values formed from these CIs will, by definition, satisfy $p_i\le t$ iff.
$\theta_{0,i}\notin CI_i(t)$.  %
If the CIs satisfy (\ref{eq:average_coverage}) for each $t\in[0,1]$, the
resulting $p$-values will
satisfy (\ref{eq:average_level_alpha}) since
$\frac{1}{m}\sum_{i\in \mathcal{H}_0} P(\theta_{0,i}\notin CI_i(t))
= \frac{1}{m}\sum_{i\in \mathcal{H}_0} P(\theta_i\notin CI_i(t))
\le \frac{1}{m}\sum_{i=1}^m P(\theta_i\notin CI_i(t))\le t$.

CIs satisfying the average coverage criterion (\ref{eq:average_coverage}) and
related criteria have been developed in a number of settings (\citealp{wahba_bayesian_1983}; \citealp{nychka_bayesian_1988};
\citealp[][Chapter 5.8]{wasserman_all_2007}; \citealp{cai_adaptive_2014};
\citealp{armstrong_robust_2022-1}).  They are particularly
appealing in high dimensional or nonparametric settings involving
regularized estimation, where impossibility results
\citep{low_nonparametric_1997} severely restrict the scope for constructing
classical tests and CIs.
Additional settings where the average significance level condition
(\ref{eq:average_level_alpha}) can be shown to hold have been
considered in recent work by
\citet{ignatiadis_empirical_2025}, \citet{ignatiadis_asymptotic_2025} and
\citet{barber_false_2025}.

We ask whether $p$-values satisfying the weaker condition
(\ref{eq:average_level_alpha}) can be used as an input to multiple
testing procedures used in the literature.
We focus on multiple testing procedures designed to control the false discovery
rate (FDR) of \citet{benjamini_controlling_1995}.
We find that average significance level control is indeed sufficient for certain
multiple testing procedures to guarantee FDR control.
In particular, average significance level control is sufficient to guarantee FDR control of the
\citet{benjamini_controlling_1995} procedure in the asymptotic regime of weakly
dependent $p$-values and many
hypotheses ($m\to\infty$) and of the \citet{benjamini_control_2001} procedure
with fixed $m$ and arbitrary dependence among $p$-values.
On the other hand, in contrast to the classical setting, we show by example that
the \citet{benjamini_controlling_1995} procedure does not in general have FDR control with
fixed $m$ and independent $p$-values,
and that
approaches that estimate the proportion of null hypotheses, such as
the procedure of \citet{storey_direct_2002}, can fail to control FDR even as
$m\to\infty$.

Much of the literature on FDR controlling multiple testing procedures takes a
family of $p$-values satisfying the
classical significance level control condition (\ref{eq:level_alpha}) as a
starting point.  %
An important exception is the literature on knockoff based FDR controlling
procedures \citep{barber_controlling_2015}, which instead rely on the construction of
auxiliary random variables, called knockoffs.  %
Constructing knockoffs typically requires modeling
assumptions such as the ``model-$X$'' framework, in which the joint distribution
of regression covariates is known or estimated with sufficient
accuracy \citep{candes_panning_2018}, or restricting the procedure to low
dimensional settings; see also \citet{arias-castro_distribution-free_2017} for
an application of this approach under the assumption of a symmetric null
distribution.

More recently, \citet{wang_false_2022} have shown that $e$-values (random
variables $e_i$ satisfying $E[e_i]\le 1$ for $i\in\mathcal{H}_0$) can be used as
an input to FDR controlling procedures, thereby providing another approach to
controlling FDR without the use of classical significance level controlling
tests.
A notion of average error control for $e$-values similar to the one used for
$p$-values in the present paper has arisen independently in this literature; see
\citet{ren_derandomised_2024}, \citet{li_note_2025} and
\citet{ignatiadis_asymptotic_2025}.
Interestingly, \citet{ren_derandomised_2024} use this idea to draw a connection
between $e$-values and the knockoff literature cited above.

While we are not aware of previous results applying the average significance
control criterion (\ref{eq:average_level_alpha}) to FDR control,
the idea of requiring coverage or size control only on average is suggested by
empirical Bayes interpretations of the FDR \citep[e.g.][]{storey_direct_2002}
and anticipated in some discussions in this literature
\citep[e.g.][]{efron_size_2007}.
Subsequent to the first draft of this paper, further results and applications
involving FDR control using the average significance level control criterion
(\ref{eq:average_level_alpha}) have been developed by
\citet{ignatiadis_empirical_2025}, \citet{ignatiadis_asymptotic_2025} and
\citet{barber_false_2025}.

The rest of this paper is organized as follows.
Section \ref{sec:overview} introduces the setup and provides an overview of results.
Section \ref{sec:finite_sample_results} presents finite sample results and their
proofs.
Section \ref{sec:asymptotic_fdr} presents results that are asymptotic in the
number $m$ of hypotheses being tested, while the proofs of the asymptotic
results are contained in the Supplementary Materials.

\section{Setup and Overview of Results}\label{sec:overview}

A multiple testing procedure is a function that maps the $p$-values
$p_1,\ldots,p_m$ to a subset
$\mathcal{R}=\mathcal{R}(p_1,\ldots,p_m)\subseteq\{1,\ldots,m\}$ of rejected
null hypotheses.
The false discovery proportion (FDP) of a procedure $\mathcal{R}$ is:
\begin{align}
  \label{eq:fdp}
  \FDP(\mathcal{R},\mathcal{H}_0)=\frac{\#\left( \mathcal{R}\cap\mathcal{H}_0 \right)}{\#\mathcal{R}\vee 1}
\end{align}
where $\#\mathcal{A}$ is the cardinality of $\mathcal{A}$ and $a\vee b$ denotes
the maximum of $a$ and $b$.
The false discovery rate (FDR) of this procedure is the expectation of the FDP:
\begin{align}
  \label{eq:fdr}
  \FDR(\mathcal{R},\mathcal{H}_0,P)
  =E_P\FDP(\mathcal{R},\mathcal{H}_0)
  =E_P\left[ \frac{\#\left( \mathcal{R}\cap\mathcal{H}_0 \right)}{\#\mathcal{R}\vee 1} \right]
\end{align}
where $E_P$ denotes expectation under the distribution $P$ of the $p$-values.
We say that $\mathcal{R}$ controls the false discovery rate at level $q$ if
$\FDR(\mathcal{R},\mathcal{H}_0,P)\le q$.

While some of our results are more general, our main focus is on
the \citet[][BH]{benjamini_controlling_1995} step-up procedure, and generalizations such as those considered by \citet{benjamini_control_2001},
\citet{storey_direct_2002} and \citet{blanchard_two_2008}.
To describe these procedures, let
\begin{align}\label{eq:fixed_t_R}
  \mathcal{R}^{\operatorname{fixed}}_t(p_1,\ldots, p_n) = \left\{ i: p_i\le t \right\}.
\end{align}
denote the fixed rejection region procedure with cutoff $t$.
That is, we reject all hypotheses with $p$-value less than $t$.
Let
\begin{align}
  \label{eq:VtStRt}
  &V(t)=\sum_{i\in\mathcal{H}_0} I(p_i\le t)
  =\#\left( \mathcal{R}^{\operatorname{fixed}}_t \cap \mathcal{H}_0 \right),
  \quad
  S(t)=\sum_{i\notin\mathcal{H}_0} I(p_i\le t)
  =\#\left( \mathcal{R}^{\operatorname{fixed}}_t \backslash \mathcal{H}_0 \right)  \nonumber  \\
  &\text{and }R(t)=V(t)+S(t)=\# \mathcal{R}^{\operatorname{fixed}}_t.
\end{align}
The FDP of $\mathcal{R}^{\operatorname{fixed}}_t$ is given by $V(t)/[R(t)\vee
1]$.
The BH procedure can be
motivated by noting that, while $V(t)$ cannot be observed, one can form a
conservative estimate by replacing it with $m\cdot t$.
This gives 
an estimate of the fixed rejection region FDR:
\begin{align}\label{eq:fdr_hat}
  \widehat{\FDR}(t) = \frac{m\cdot t}{\# \mathcal{R}^{\operatorname{fixed}}_t\vee 1}
  = \frac{m\cdot t}{R(t)\vee 1}.
\end{align}
The BH procedure at nominal FDR level $q$ uses a cutoff $\hat
t_{\operatorname{BH},q}$ based on this estimate:
\begin{align}
  \label{eq:bhq}
  \mathcal{R}_{\operatorname{BH},q}(p_1,\ldots,p_m)
  = \{i: p_i\le \hat t_{\operatorname{BH},q}\}
  \quad\text{where}\quad
  \hat t_{\operatorname{BH},q} = \max\{t : \widehat{\FDR}(t)\le q\}.
\end{align}
A more general class of step-up procedures can be formed by using an estimate of
the form $\pi m t$
for $V(t)$ and modifying the
denominator using a nondecreasing function $\beta$, called a shape function:
\begin{align}
  \label{eq:step-up}
  \mathcal{R}_{\pi,\beta(\cdot),q}(p_1,\ldots,p_m)
  = \{i: p_i\le \hat t_{\pi,\beta(\cdot),q}\}
  \text{ where }
  \hat t_{\pi,\beta(\cdot),q} = \max\left\{t : \frac{\pi m t}{\beta(R(t))}\le q\right\}.
\end{align}
Such procedures have been considered by, among others, \citet{benjamini_control_2001},
\citet{storey_direct_2002} and \citet{blanchard_two_2008}.

When the $p$-values satisfy the classical significance level control condition
(\ref{eq:level_alpha}), these procedures are known to have the following
properties.
\begin{itemize}
\item[(i)] The BH procedure controls FDR when $p$-values are
  independent \citep{benjamini_controlling_1995}.
\item[(ii)] The estimate $\widehat{\FDR}(t)$ is upwardly biased for the FDR of
  the fixed rejection region procedure $\mathcal{R}^{\operatorname{fixed}}_t$
  when $p$-values are independent \citep{storey_strong_2004,liang_adaptive_2012}.
\item[(iii)] The procedure $\mathcal{R}_{1,\beta(\cdot),q}$ (with $\pi=1$) controls FDR
  under arbitrary dependence for the shape function $\beta(k)=k\left( \sum_{i=1}^m i^{-1}
\right)^{-1}$ \citep{benjamini_control_2001} and, more generally, when
$\beta(k)=\int_0^k xd\nu(x)$ for an arbitrary probability distribution $\nu$ on
$(0,\infty)$ \citep{blanchard_two_2008}.
\item[(iv)] The BH procedure controls FDR asymptotically (as $m\to\infty$) when the
$p$-values satisfy a weak dependence condition
\citep{storey_strong_2004,genovese_stochastic_2004}.
\item[(v)] The procedure $\mathcal{R}_{\hat\pi, \beta(\cdot), q}$, where
  $\beta(t)=t$ and
$\hat \pi=(\sum_{i=1}^mI(p_i>\lambda)+1)/((1-\lambda) m)$ is an
estimate of $\#\mathcal{H}_0/m$, controls FDR (a) under fixed $m$ with independent
$p$ values
using a slight modification of the procedure
\citep{storey_strong_2004}\footnote{For results with fixed $m$,
  \citet{storey_strong_2004} consider a modification in which $\hat
  t_{\pi,\beta(\cdot),q}$ is replaced by $\lambda$ if $\hat t_{\pi,\beta(\cdot),q}>\lambda$.} and (b) asymptotically as $m\to\infty$
when the $p$-values satisfy a weak dependence condition
\citep{storey_strong_2004,genovese_stochastic_2004}.
\end{itemize}
Our results can be summarized as showing that, when the $p$-values only satisfy
the weaker average significance level control condition
(\ref{eq:average_level_alpha}), properties (ii), (iii) and (iv) continue to hold,
but that properties
(i) and (v)(a) and (v)(b) in general do not.
Section \ref{sec:finite_sample_fdr} shows property (iii) and provides a
counterexample to property (i).
Section \ref{sec:fdr_estimation} shows property (ii).
Section \ref{sec:asymptotic_fdr} shows property (iv).
A counterexample for property (v) is given in the Supplementary Materials.

\section{Finite Sample Results}\label{sec:finite_sample_results}

This section considers finite sample control of FDR for step-up procedures
(Section \ref{sec:finite_sample_fdr}) and point estimation of FDR of the fixed
rejection region procedure $\mathcal{R}^{\operatorname{fixed}}_t$ (Section \ref{sec:fdr_estimation}).

\subsection{FDR Control}\label{sec:finite_sample_fdr}

Our result on FDR control for
step-up procedures is a corollary of a more general result that uses an
invariance assumption on an oracle version of a multiple testing procedure.
The basic idea is that, if the $p$-values satisfy the average significance level
control condition (\ref{eq:average_level_alpha}), then one can form another
multiple testing problem in which the classical condition (\ref{eq:level_alpha})
holds by randomly permuting the $p$-values of the true null hypotheses and
multiplying them by $m/\#\mathcal{H}_0$.
One can then apply results from the literature to this new setting.
To state our result, we explicitly introduce notation
$\mathcal{R}(p_1,\ldots,p_m;\mathcal{H}_0)$
for
oracle procedures that depend on the set of true null
hypotheses $\mathcal{H}_0$ (typically through the cardinality $\#\mathcal{H}_0$
of this set).
We use
a permutation invariance
condition
\begin{align}
  \label{eq:permutation_invariance}
  i\in \mathcal{R}(p_1,\ldots,p_m) \quad\text{iff.}\quad \sigma(i)\in \mathcal{R}(p_{\sigma(1)},\ldots,p_{\sigma(m)})
\end{align}
for any permutation $\sigma$ of the indices $1,\ldots,m$ of the tests.
This includes the class of step-up procedures (\ref{eq:step-up}), so long as
$\pi$ is either a fixed number or a permutation invariant function of the
$p$-values.

\begin{theorem}\label{thm:finite_sample_fdr_control}
  Let $\mathcal{R}$ be a multiple testing procedure that satisfies the
  permutation invariance condition (\ref{eq:permutation_invariance}), and
  suppose that the oracle procedure
  $\widetilde{\mathcal{R}}(p_1,\ldots,p_m;\mathcal{H}_0)\allowbreak =\mathcal{R}(p_1(m_0/m),\ldots,p_m(m_0/m))$
  (where $m_0=\#\mathcal{H}_0$) controls FDR at level $q$ for any
  $(P,\mathcal{H}_0)$ satisfying the classical significance level control
  condition (\ref{eq:level_alpha}),
  regardless of the dependence structure of $p_1,\ldots,p_m$ under $P$.
  Then $\mathcal{R}$ controls FDR at level
  $q$ for any $(P,\mathcal{H}_0)$ such that the average significance level control
  condition (\ref{eq:average_level_alpha}) holds.
\end{theorem}
\begin{proof}
  Given $(P,\mathcal{H}_0)$ such that (\ref{eq:average_level_alpha}) holds and
  $p_1,\ldots,p_n$ drawn from $P$, define $\tilde p_i$ as follows.  Let $\sigma$
  be a permutation of $\mathcal{H}_0$, taken at random from the set of all
  permutations of $\mathcal{H}_0$ with equal probability, independently of $p_1,\ldots,p_m$.  Extend $\sigma$ to a
  permutation on $\{1,\ldots,m\}$ by taking $\sigma(i)=i$ for $i\notin
  \mathcal{H}_0$.  Let $\tilde p_i=(m/m_0)p_{\sigma(i)}$, where
  $m_0=\#\mathcal{H}_0$.  Then, for $i\in\mathcal{H}_0$ and $t\in[0,1]$,
  \begin{align*}
    P(\tilde p_i\le t) = \sum_{j\in\mathcal{H}_0} P(\sigma(i)=j)P(p_j(m/m_0)\le t|\sigma(i)=j)
    =\frac{1}{m_0}\sum_{j\in\mathcal{H}_0} P(p_j(m/m_0)\le t)
  \end{align*}
  where we use independence of $\sigma$ and $p_j$ and the fact that
  $P(\sigma(i)=j)=1/m_0$.  Since $p_1,\ldots,p_m$ satisfy
  (\ref{eq:average_level_alpha}) under $(P,\mathcal{H}_0)$, this is bounded by
  $(m/m_0)\cdot t m_0/m=t$.  Thus, letting $\tilde P$ denote the
  distribution of $\tilde p_1,\ldots,\tilde p_n$ under $P$, $(\tilde
  P,\mathcal{H}_0)$ satisfies the classical significance level
  control condition (\ref{eq:level_alpha}).  It follows by the assumptions of
  the theorem that the oracle procedure
  $\widetilde{\mathcal{R}}(\tilde p_1,\ldots,\tilde
  p_m;\mathcal{H}_0)=\mathcal{R}(\tilde p_1(m_0/m),\ldots,\tilde p_m(m_0/m))=\mathcal{R}(p_{\sigma(1)},\ldots,p_{\sigma(m)})$
  controls FDR at level $q$ under $\mathcal{H}_0$ when $p_1,\ldots,p_m$ are
  drawn according to $P$.  But by permutation invariance of $\mathcal{R}$ and
  the fact that $\sigma$ maps $\mathcal{H}_0$ to itself, we
  have $\#\left( \mathcal{R}(p_{\sigma(1)},\ldots,p_{\sigma(m)}) \cap
    \mathcal{H}_0 \right) = \#\left( \mathcal{R}(p_{1},\ldots,p_{m}) \cap
    \mathcal{H}_0 \right)$.
  Also, $\# \mathcal{R}(p_{\sigma(1)},\ldots,p_{\sigma(m)})=\#
  \mathcal{R}(p_{1},\ldots,p_{m})$ by permutation invariance.  Thus, the FDR of
  $\mathcal{R}(p_{m},\ldots,p_{m})$ is the same as the FDR of
  $\mathcal{R}(p_{\sigma(1)},\ldots,p_{\sigma(m)})$, and is therefore bounded by
  $q$.
\end{proof}

As a special case, applying Proposition 2.7 and Lemma 3.2(iii) in
\citet{blanchard_two_2008} gives the following.

\begin{corollary}\label{cor:by_br_fdr_control}
  The class of dependence controlling step-up procedures of
  \citet{blanchard_two_2008}, given by (\ref{eq:step-up}) with $\pi=1$ and
  $\beta(r)=\int_{0}^{r}xd\nu(x)$ for some probability measure $\nu$, controls
  FDR at level $q$ for any $(P,\mathcal{H}_0)$ such that the average
  significance level control condition (\ref{eq:average_level_alpha}) holds.
  In particular, 
  the step-up procedure of \citet{benjamini_control_2001}, which is given by
  (\ref{eq:step-up}) with $\pi=1$ and $\beta(r)=r/\left( \sum_{i=1}^m 1/i
  \right)$, controls FDR at level $q$ for any $(P,\mathcal{H}_0)$ such that the average
  significance level control condition (\ref{eq:average_level_alpha}) holds.
\end{corollary}

Key requirements here are that the original procedure (a) controls FDR under
arbitrary dependence and (b) can incorporate the $m/m_0$ adjustment through an
oracle result.
In particular, (b) rules out procedures of the form
$\mathcal{R}_{\hat\pi,\beta(\cdot),q}$ with $\hat\pi$ an estimate of $m_0/m$, as
in \citet{storey_direct_2002}.
Clearly, ruling out estimates of $m_0/m$ is necessary, since
such estimates attempt to use a bound $m_0\cdot t$ on $V(t)$, whereas
average coverage only gives a bound of
$m\cdot t$ on the expectation of $V(t)$ (see the Supplementary Materials for a
counterexample).
Regarding (a), note that even if the
$p$-values satisfy some dependence structure that would guarantee FDR control
under classical significance level control (e.g. independence or the positive
regression dependency on a subset condition used by
\citet{benjamini_control_2001}), FDR control is not guaranteed.
The following counterexample shows that (a) is necessary in general.
In
particular, the BH procedure need not control FDR under the average significance
level control condition (\ref{eq:average_level_alpha}) and independent
$p$-values.

Suppose $m\ge 2$ and $q<2/3$.
Let $P(p_1\le t)=t\cdot m$ for $0\le t\le (3/2)\cdot (q/m)$
and let $P(p_1\in ((3/2)\cdot (q/m),2q/m])=0$.
Let
$P(p_2\in [a,b])=(b-a)\cdot m$ for any $(3/2)\cdot (q/m)\le a\le b\le 2q/m$
and let $P(p_2\in [0,(3/2)\cdot (q/m)))=0$.
We
can then distribute the remaining probability mass of $p_1,p_2$ and
$p_3,\ldots,p_m$ over the set $(2q/m,1]$ so that $\frac{1}{m}\sum_{i=1}^m
P(p_i\le t)\le t$ for all $t\in [0,1]$ (for example, we can set $p_3,\ldots,p_m$
to be equal to $1$ with probability one, and we can set the remaining
probability mass for $p_1$ and $p_2$ to point masses at $1$).
Thus, the average
significance level condition
(\ref{eq:average_level_alpha}) holds, with $\mathcal{H}_0=\{1,\ldots,m\}$.  Now
consider the FDR of the Benjamini-Hochberg procedure, which rejects all
hypotheses $i$ such that $p_i\le q\hat r/m$ where $\hat r$ is the number of
rejected hypotheses.  The FDR is equal to the probability of at least one rejection
in this case (since $\mathcal{H}_0=\{1,\ldots,m\}$).  Note that the event
$p_1\le q/m$ implies that hypothesis $1$ is rejected, and this has probability
$q$.  But the event $q/m< p_1\le (3/2)\cdot (q/m)$ and $p_2\le 2q/m$ has
probability $(q/2)\cdot (q/2)$, and it is disjoint
with the event $p_1\le q/m$.  This gives a lower bound of $q+[(q/m)\cdot(1/2)]^2>q$ for
the FDR.  Thus, the FDR is not controlled at level $q$.

It is worth mentioning here some further results that have been obtained
subsequent to the first draft of this paper.
First,
\citet{ignatiadis_asymptotic_2025} have developed results for $e$-values using a
related notion of average error control: they show that if the $e$-values
$e_1,\ldots,e_m$ satisfy $(1/m)\sum_{i\in\mathcal{H}_0}E[e_i]\le 1$, then the
$e$-BH procedure of \citet{wang_false_2022} controls FDR at the nominal level.
Using this result and a certain calibration of $p$-values to $e$-values, they
provide an alternative proof of the
main conclusion of Corollary \ref{cor:by_br_fdr_control}.
Second,
\citet{barber_false_2025} have obtained several results characterizing the FDR
properties of the BH procedure under (\ref{eq:average_level_alpha}) and various
dependence assumptions on the $p$-values.  In particular, they show that, when
$p$-values are independent, applying the procedure %
(\ref{eq:step-up}) with $\pi=1$ and $\beta(r)=r/1.93$ (i.e. the BH procedure at
nominal level $q/1.93$) controls FDR.  This provides a much less conservative
procedure compared to those in Corollary \ref{cor:by_br_fdr_control} in the case
of independent $p$-values.

\subsection{Estimation of FDR for Fixed Rejection Region}\label{sec:fdr_estimation}

We now consider using the BH cutoff as an estimate of the FDR for a fixed
rejection region multiple testing procedure.
Under independent $p$-values, it is known that $\widehat{\FDR}(t)$ is an
upwardly biased estimate of $\FDR(\mathcal{R}^{\operatorname{fixed}}_t)$ under the
classical significance level control condition (\ref{eq:level_alpha})
(see \citet{storey_strong_2004} and the correction by
\citet{liang_adaptive_2012}).
Indeed, a slightly weaker assumption of independence between null and non-null
$p$-values suffices.\footnote{I thank an anonymous referee for pointing this out.}
We now show that this property continues to hold
under the weaker average significance level control condition
(\ref{eq:average_level_alpha}).
The result essentially follows from the same
arguments as in the case where the $p$-values satisfy the classical
significance level control condition.

\begin{theorem}
  Suppose that $(P,\mathcal{H}_0)$ satisfies the average significance level control condition
  (\ref{eq:average_level_alpha}) and that the null $p$-values
  $\{p_i\}_{i\in\mathcal{H}_0}$ are statistically independent of the non-null
  $p$-values $\{p_i\}_{i\notin\mathcal{H}_0}$.  Then
  $E_P\widehat{\FDR}(t)\ge \FDR(\mathcal{R}^{\operatorname{fixed}}_t,\mathcal{H}_0,P)$.
\end{theorem}
\begin{proof}
  For $V(t)$ and $S(t)$ defined in (\ref{eq:VtStRt}), we have
\begin{align*}
  E_P\widehat{\FDR}(t) - \FDR(\mathcal{R}^{\operatorname{fixed}}_t,\mathcal{H}_0,P)
  =E_P\frac{m\cdot t-V(t)}{[V(t)+S(t)]\vee 1}
  \ge E_P\frac{m\cdot t-V(t)}{[m\cdot t+S(t)]\vee 1}
\end{align*}
(the last step follows by noting that replacing $V(t)$ with $m\cdot t$ in the denominator weakly decreases
the denominator when the numerator is negative and weakly increases the
denominator when the numerator is positive).
The result then follows by noting that $S(t)$ and $V(t)$ are independent by the
independence assumption on $p$-values, and that $E_PV(t)\le m\cdot t$ by the
assumption that the $p$-values satisfy the average significance level control
condition (\ref{eq:average_level_alpha}).
\end{proof}

\section{Asymptotic Results}\label{sec:asymptotic_fdr}

We now consider asymptotic FDR control, under a sequence $P=P^{(m)}$ of
probability measures and $\mathcal{H}_0=\mathcal{H}_0^{(m)}$ and $m\to\infty$.
We suppress the dependence on $m$ whenever it doesn't cause confusion, but we
note that the $p$-values form a triangular array, since the distribution (and
the set $\mathcal{H}_0$ of true null hypotheses) can change with $m$.
Recall the definitions of $V(t)$, $S(t)$
and $R(t)$ in (\ref{eq:VtStRt}).
If the average significance level control condition
(\ref{eq:average_level_alpha}) holds, and the $p$-values do not exhibit too
much statistical dependence,
we will have
\begin{align}
  \label{eq:asymptotic_average_level}
  \frac{1}{m}V(t)
  \le t+o_P(1)
  \text{ for all }t\in [0,1].
\end{align}
For some results, we also assume a law
of large numbers for the total rejections and rejected true nulls:
\begin{align}
  \label{eq:total_rejections_lln}
  \frac{1}{m}V(t)
  \stackrel{p}{\to} G(t)\le t
  \quad\text{and}\quad
  \frac{1}{m}R(t)
  \stackrel{p}{\to} F(t)
  \text{ for all }t\in [0,1].
\end{align}
These assumptions are analogous to assumptions made for asymptotic FDR control
under classical significance level control in the literature \citep[e.g.][Eq.
(7)-(9)]{storey_strong_2004}.  The difference here is that the conditions are
weaker, since the upper bound in (\ref{eq:asymptotic_average_level}) is given by
$t$ rather than $t\pi_0$ where $\pi_0$ is the limit of $\#\mathcal{H}_0/m$.
As one might expect, this will lead to problems for ``adaptive'' procedures that
attempt to estimate $\pi_0$ (see the Supplementary Materials for a
counterexample).
However, as we now show, it is not a
problem for the Benjamini-Hochberg procedure, which uses the conservative upper
bound of $1$.
We first show conservative consistency of the BH cutoff (\ref{eq:fdr_hat}) for
the FDR (and FDP) of the fixed rejection region procedure.

\begin{theorem}\label{thm:conservative_consistency}
  Let $\widehat{\FDR}(t)$ be the BH estimate, given in (\ref{eq:fdr_hat}), of
  the FDR of the fixed rejection region procedure $\mathcal{R}^{\text{fixed}}_t$
  given in (\ref{eq:fixed_t_R}) and suppose that
  (\ref{eq:asymptotic_average_level}) holds.  Let $\underline t$ be such that
  there exists $\eta>0$ with $\frac{1}{m}\sum_{i=1}^m I(p_i\le \underline t)\ge
  \eta +  o_P(1)$.
  Then
  \begin{align*}
    \inf_{t\in [\underline t, 1]}\left[ \widehat{\FDR}(t)-\FDP(\mathcal{R}^{\text{fixed}}_t,\mathcal{H}_0) \right]
    \ge o_P(1).
  \end{align*}
  If, in addition, (\ref{eq:total_rejections_lln}) holds for continuous
  functions $G$ and $F$, then, letting $\FDR_\infty(t)=G(t)/F(t)$, we have
  \begin{align*}
    &\sup_{t\in [\underline t, 1]} \left| \FDP(\mathcal{R}^{\text{fixed}}_t,\mathcal{H}_0) - \FDR_\infty(t) \right|\stackrel{p}{\to } 0,
    \quad
      \sup_{t\in [\underline t, 1]} \left| \FDR(\mathcal{R}^{\text{fixed}}_t,\mathcal{H}_0,P) - \FDR_\infty(t) \right|\to 0  \\
    &\text{ and }
    \inf_{t\in [\underline t, 1]}\left[ \widehat{\FDR}(t)-\FDR(\mathcal{R}^{\text{fixed}}_t,\mathcal{H}_0,P) \right]
     \ge o_P(1).
  \end{align*}
\end{theorem}

The proof of Theorem \ref{thm:conservative_consistency} is given in the
Supplementary Materials.
Next, we state a result showing asymptotic control of FDR for the BH procedure
$\mathcal{R}_{\operatorname{BH},q}$ defined in (\ref{eq:bhq}).

\begin{theorem}\label{thm:asymptotic_fdr_control}
  Suppose Assumptions (\ref{eq:asymptotic_average_level}) and
  (\ref{eq:total_rejections_lln}) hold for continuous functions $G$ and $F$ and
  that there exists $t^*>0$ such that $F(t^*)>0$ and $G(t^*)/F(t^*)<q$.  Then
  \begin{align*}
    \FDP(\mathcal{R}_{\operatorname{BH},q},\mathcal{H}_0)\le q + o_P(1)
    \quad\text{and}\quad
    \FDR(\mathcal{R}_{\operatorname{BH},q},\mathcal{H}_0,P)\le q + o(1).
  \end{align*}
\end{theorem}

The proof of Theorem \ref{thm:asymptotic_fdr_control} is given in the
Supplementary Materials.
The condition on $t^*$ used in Theorem \ref{thm:asymptotic_fdr_control} imposes
a lower bound on the proportion of total rejections relative to null rejections
at nominal level $t^*$.  This condition requires that
the hypothesis tests have sufficient power on average, and that the proportion
$(m-\#\mathcal{H}_0)/m$ of non-null hypotheses is not too small
as $m\to\infty$.
Similar conditions have been used in the asymptotic analysis of multiple
hypothesis testing procedures under classical significance level control
\citep[e.g.][Theorem 4]{storey_strong_2004}.

\bibliography{../../../../../library.bib}

\newpage

\appendix
\renewcommand{\theequation}{S\arabic{equation}}
\renewcommand{\thelemma}{S\arabic{lemma}}

\section{Supplementary Material}

This supplementary material contains proofs of Theorems 3 and 4 from the main
text, as well as a counterexample showing that procedures that use an estimate
of the proportion of true null hypotheses may not control FDR when applied to
average significance level controlling $p$-values.

\subsection{Proof of Theorem \ref{thm:conservative_consistency}}

We begin with a lemma regarding uniform convergence.

\begin{lemma}\label{lemma:uniform_convergence}
  If (\ref{eq:asymptotic_average_level}) holds, then
  \begin{align}
  \label{eq:uniform_t_asymptotic_size}
  \inf_{t\in [0,1]}\left[ t - V(t)/m \right]\ge o_P(1).
  \end{align}
  If (\ref{eq:total_rejections_lln}) holds for continuous functions $F$ and $G$,
  then
  \begin{align}
    \label{eq:uniform_t_lln}
    \sup_{t\in [0,1]}\left| V(t)/m - G(t)  \right|\stackrel{p}{\to } 0
    \quad\text{and}\quad
    \sup_{t\in [0,1]}\left| R(t)/m - F(t)  \right|\stackrel{p}{\to } 0.
  \end{align}
\end{lemma}
\begin{proof}
The result follows by arguing as in the proof of the Glivenko-Cantelli Theorem.
Let
$\overline t_K(t)$ be the least element in $\{0,1/K,\ldots,(K-1)/K,1\}$ that is greater
than or equal to $t$.
To show (\ref{eq:uniform_t_asymptotic_size}), note that
$t-V(t)/m\ge t-V(\overline t_K(t))/m\ge \overline t_K(t)-V(\overline t_K(t)/m)-1/K$.
Thus,
$\inf_{t\in [0,1]}\left[ t - V(t)/m \right]\ge \min_{t\in
  \{0,1/K,\ldots,(K-1)/K,1\}}\left[ t - V(t)/m \right]-1/K$.  Applying
(\ref{eq:asymptotic_average_level}) and taking $K\to\infty$ gives
(\ref{eq:uniform_t_asymptotic_size}).

Similarly, to show that (\ref{eq:uniform_t_lln}) holds so long as
(\ref{eq:total_rejections_lln}) holds for continuous $G$ and $F$, note that
  $V(t)/m - G(t) \le V(\overline t_K(t))/m - G(t)
  \le V(\overline t_K(t))/m - G(\overline t_K(t))
  +\omega(1/K)$
where $\omega(\varepsilon)=\sup_{s,t\in [0,1], |s-t|\le \varepsilon}
|G(s)-G(t)|$ satisfies $\lim_{\varepsilon\to 0}\omega(\varepsilon)=0$ by
uniform continuity of $G$ on the interval $[0,1]$.  From this and an analogous
lower bound, it follows that
$\sup_{t\in [0,1]} |V(t)/m-G(t)|\le
\max_{t\in \{0,1/K,\ldots,(K-1)/K,1\}}|V(t)/m-G(t)| + \omega(1/K)$.  Applying
(\ref{eq:total_rejections_lln}) and taking $K\to 0$ gives the first part of
(\ref{eq:uniform_t_lln}).  The same arguments applied to $R(t)/m-F(t)$ give the
second part of (\ref{eq:uniform_t_lln}).

\end{proof}

We now prove Theorem \ref{thm:conservative_consistency}.
We begin by showing that the first claim in the theorem holds under
(\ref{eq:asymptotic_average_level}).
For any $\varepsilon>0$,
the event
$\inf_{t\in [\underline t, 1]}\left[\widehat{\FDR}(t)-\FDP(\mathcal{R}^{\text{fixed}}_t,\mathcal{H}_0)\right]<-\varepsilon$
implies that there exists $t\in [\underline t, 1]$ such that $t-V(t)/m<-\varepsilon(R(t)\vee 1)/m\le -\varepsilon (R(\underline t)\vee 1)/m$.  This
implies $\inf_{t\in [\underline t, 1]}[t-V(t)/m]\le -\varepsilon\cdot\eta/2$
on the event $R(\underline t)/m=\frac{1}{m}\sum_{i=1}^mI(p_i\le \underline
t)\ge\eta/2$, which holds
with probability approaching one by assumption.
The first claim of Theorem \ref{thm:conservative_consistency} now follows by
noting that the probability of this event converges to zero by (\ref{eq:uniform_t_asymptotic_size}).

Next, we show that the second claim of the theorem holds if
(\ref{eq:asymptotic_average_level}) holds and 
(\ref{eq:total_rejections_lln}) holds for some continuous $F$ and $G$.
The first part of the second claim follows immediately from
(\ref{eq:uniform_t_lln}), using uniform continuity of the function $(a,b)\mapsto
a/b$ over $b\in [F(\underline t),1]$ since $F(\underline t)\ge \eta>0$.
The second part of the second claim then follows immediately from the dominated
convergence theorem.
The third part of the second claim follows from combining the first claim in the theorem
with the first and second part of the second claim in the theorem.

\subsection{Proof of Theorem \ref{thm:asymptotic_fdr_control}}

We have
\begin{align*}
    \FDP(\mathcal{R}_{\operatorname{BH},q},\mathcal{H}_0)
    = \frac{V(\hat t_{\operatorname{BH},q})/m}{[R(\hat t_{\operatorname{BH},q})\vee 1]/m}
    \le \widehat{\FDR}(\hat t_{\operatorname{BH},q}) + I(\hat t_{\operatorname{BH},q}< t^*) + o_P(1)
\end{align*}
  using the fact that $\sup_{t\in [t^*,1]}\left[ \frac{V(t)/m}{[R(t)\vee 1]/m} -
    \widehat{\FDR}(t) \right]\le o_P(1)$ by Theorem \ref{thm:conservative_consistency}.
  Since $\widehat{\FDR}(\hat t_{\operatorname{BH},q})\le q$ by
  construction, it suffices to show $P(\hat t_{\operatorname{BH},q}\ge
  t^*)\to 1$.  But this follows since $\widehat{\FDR}(t^*)\le q$ implies
  $\hat t_{\operatorname{BH},q}\ge t^*$, and
  $\widehat{\FDR}(t^*)\stackrel{p}{\to} G(t^*)/F(t^*)<q$ by
  (\ref{eq:total_rejections_lln}).
  This shows that $\FDP(\mathcal{R}_{\operatorname{BH},q},\mathcal{H}_0)\le
  q+o_P(1)$, from which it also follows that $\FDR(\mathcal{R}_{\operatorname{BH},q},\mathcal{H}_0,P)\le
  q+o(1)$ by dominated convergence.

\subsection{Counterexample for adaptive procedures}

As noted in the main text, procedures that attempt to incorporate
estimates $\hat\pi$ of $\#\mathcal{H}_0/m$ will in not, in general, lead finite
sample or asymptotic FDR control under weak dependence conditions on the
$p$-values.
As a counterexample, consider the procedure
  $\mathcal{R}_{\hat\pi,\beta(\cdot),q}$ with $\beta(t)=t$ that uses the estimate
$\hat\pi=(\sum_{i=1}^m I(p_i>\lambda)+1)/((1-\lambda) m)$ of $\pi_0$ for some
$\lambda\in(0,1)$ as in \citet{storey_strong_2004} and suppose that
$\#\mathcal{H}_0/m\to \pi_0$ where $q<\pi_0<\lambda$.
Suppose that $p_i$ is
independent over $i$ and
follows a uniform distribution on $[0,\pi_0]$ for $i\in\mathcal{H}_0$, and that
$p_i$ follows some other distribution with support contained in $[0,\pi_0)$ for
$i\not\in\mathcal{H}_0$.  Then (\ref{eq:asymptotic_average_level}) and
(\ref{eq:total_rejections_lln}) hold.  However, note that
the cutoff $\hat t_{\hat\pi,\beta(\cdot),q}$ in (\ref{eq:step-up}) will be equal
to $1$ so long as $q\ge \hat\pi m \cdot 1/R(1)=\hat \pi$, which holds for large
enough $m$ with probability one since $\hat\pi=1/((1-\lambda) m)\to 0$ with
probability one.
Since the FDP is equal to $\#\mathcal{H}_0/m$ on the event that $\hat
t_{\hat\pi,\beta(\cdot),q}=1$, this implies that the FDP is equal to
$\#\mathcal{H}_0/m\to \pi_0>q$ with probability one for large enough $m$, so
that the FDR is equal to $\#\mathcal{H}_0/m>q$ for large enough $m$.
Note that the FDP is also equal to $\#\mathcal{H}_0/m$ on the event that $\hat
t_{\hat\pi,\beta(\cdot),q}\ge \lambda$, since the $p$-values are all supported
on $[0,\lambda)$.  Thus, the same argument goes through for the modification of
the procedure in which $\hat t_{\hat\pi,\beta(\cdot),q}$ is replaced by
$\lambda$ when $\hat t_{\hat\pi,\beta(\cdot),q}>\lambda$.

Note also that the oracle version of this procedure
($\mathcal{R}_{\pi_0,\beta(\cdot),q}$ with $\beta(t)=t$ and $\pi$ set to the true
asymptotic null proportion $\pi_0$) also fails to control FDR asymptotically.
For this procedure, the cutoff $\hat t_{\pi_0,\beta(\cdot),q}$ is equal to $1$
so long as $q\ge \pi_0$, which holds by assumption.  Thus, this oracle procedure
has FDP equal to $\pi_0$ with probability one, which means that it has FDR equal
to $\pi_0$ for all $m$.

\end{document}